\documentclass[journal,onecolumn]{IEEEtran}

\usepackage{amsmath,amssymb,amsthm,booktabs}
\usepackage[shortlabels]{enumitem}
\usepackage{cite}

\newtheorem{thm}{Theorem}
\newtheorem{lem}[thm]{Lemma}
\newtheorem{prop}[thm]{Proposition}
\newtheorem{cor}[thm]{Corollary}

\newtheorem{remark}{Remark}

\newcommand{\Ff}{{\mathbb F}}
\newcommand{\cc}{{\mathcal C}}
\newcommand{\cD}{{\mathcal D}}

\newcommand{\dsb}[1]{[\![#1]\!]}

\begin{document}

\title{Quantum Locally Repairable Codes from Negacyclic and Repeated-Root Cyclic Codes over Small Fields}

\author{Ruipan~Yang, Qiang~Fu, and~Liangdong~Lu% <-this % stops a space
\thanks{The authors are with Air Force Engineering University, Xi'an, China. e-mail: yangruipan@aliyun.com.}% <-this % stops a space
\thanks{Manuscript received XX XX, 2026; revised XX XX, 2026.}}

% The paper headers
\markboth{Submitted for publication}%
{Yang \MakeLowercase{\textit{et al.}}: Quantum Locally Repairable Codes from Negacyclic and Repeated-Root Cyclic Codes over Small Fields}

\maketitle

\begin{abstract}
Quantum locally recoverable codes (qLRCs), introduced recently by Golowich and Guruswami, allow any single-qudit erasure to be recovered from a small set of other qudits. Most known constructions require a large alphabet. We systematically investigate qLRCs obtained, via the CSS construction, from classical constacyclic codes over small fields $\Ff_q$ with $q\in\{2,3,4,5,7\}$. First, we prove that a nonzero dual-containing $\lambda$-constacyclic code exists only when $\lambda^2=1$, so that negacyclic and (repeated-root) cyclic codes exhaust the constacyclic route to qLRCs. Second, we show that the locality of a constacyclic code equals the minimum distance of its dual minus one, and we give a simple purity criterion for the resulting quantum codes. Third, we show that odd-like duadic codes whose splitting is given by $\mu_{-1}$ yield pure qLRCs; specializing to $q$-ary quadratic residue codes of prime length $p\equiv 3 \pmod 4$ gives an infinite family of pure qLRCs with unbounded minimum distance and certified locality. Finally, by means of concrete computations, we obtain a classification of qLRCs from cyclic, negacyclic, and repeated-root cyclic codes of moderate lengths, which contains the first binary qLRCs from repeated-root cyclic codes and many parameter sets that cyclic codes cannot attain.
\end{abstract}


\IEEEpeerreviewmaketitle

\section{Introduction}

\IEEEPARstart{C}{lassical} locally repairable codes (LRCs) \cite{Gopalan2012} are a cornerstone of modern distributed storage: an $[n,k,d]_q$ linear code has \emph{locality} $r$ if every code symbol can be recovered from at most $r$ other symbols. The parameters of an LRC with locality $r$ obey the Singleton-like bound
\begin{equation}\label{C-Singleton}
d\le n-k-\left\lceil\frac{k}{r}\right\rceil+2,
\end{equation}
and a large body of work provides optimal or near-optimal constructions; see, e.g., \cite{Tamo2014,Jin2019,Chen2018,Luo2019} and the references therein.

Quantum locally recoverable codes (qLRCs) were introduced by Golowich and Guruswami \cite{Golowich2023} as the quantum analogue of LRCs, motivated by large-scale quantum storage and by connections to quantum LDPC codes. An $\dsb{n,\kappa,\delta}_q$ qLRC with locality $r$ is a quantum error-correcting code in which every erased qudit can be recovered by a recovery channel acting on at most $r$ other qudits. In the stabilizer formalism, locality is certified by the existence of low-weight stabilizers anchored at each coordinate. Building on the CSS construction, Luo, Chen, Ezerman, and Ling \cite{Luo2023} derived bounds on qLRCs from their classical ingredients, including the quantum Singleton-like bound
\begin{equation}\label{Q-Singleton}
2\delta\le n-\kappa-2\left\lceil\frac{\kappa}{r}\right\rceil+4,
\end{equation}
and characterized the optimal pure qLRCs attaining \eqref{Q-Singleton} with equality. Galindo et al.~\cite{Galindo2024} introduced quantum $(r,\delta)$-LRCs and gave a necessary and sufficient condition for local recoverability of stabilizer codes.

Known explicit qLRC constructions \cite{Golowich2023,Luo2023,Galindo2024,Sharma2024,Li2025,Rajpurohit2026} predominantly require a large field size $q$ (GRS-type, algebraic-geometry, matrix-product, or good-polynomial ingredients). By contrast, small-alphabet qLRCs---which are the most relevant ones for physical qubit or qutrit systems---remain scarce; to the best of our knowledge, the only systematic binary families appear in the very recent work \cite{Stylianou2026}, which is based on subset-inclusion matrices rather than algebraic code families. We note that duadic constacyclic codes over $\Ff_4$ were used in \cite{Dastbasteh2023} to construct binary quantum stabilizer codes with growing minimum distance, but \emph{without any local repair property}; quantum codes from duadic codes more generally, e.g., \cite{Aly2007}, likewise do not address locality. To our knowledge, the present work is the first to use duadic and QR codes for quantum \emph{locally repairable} codes.

In this paper, we construct quantum locally repairable codes from cyclic, negacyclic, and repeated-root cyclic codes over small fields. Our work is as follows.
\begin{enumerate}
\item \textbf{A classification theorem} (Theorem~\ref{thm:lambda}): a nonzero dual-containing $\lambda$-constacyclic code exists only when $\lambda^2=1$. This extends the self-dual case settled by Blackford \cite[Corollary~2]{Blackford2013} to the dual-containing condition required by the CSS construction. Hence the constacyclic route to qLRCs is exhausted by cyclic and negacyclic codes, including repeated-root ones.
\item \textbf{Two structural lemmas} (Lemmas~\ref{lem:locality} and~\ref{lem:pure}): for a constacyclic code $\cc$, the locality equals $d(\cc^{\perp})-1$; and the CSS quantum code from a dual-containing $\cc$ with $d(\cc)<d(\cc^\perp)$ is pure with $\delta=d(\cc)$. Together they reduce the construction of pure qLRCs to the computation of $d(\cc)$ and $d(\cc^\perp)$.
\item \textbf{An infinite family from duadic codes} (Theorem~\ref{thm:duadic}): every odd-like duadic code whose splitting is given by $\mu_{-1}$ and satisfying $d<d^\perp$ yields a pure qLRC with locality $r=d^\perp-1$. In particular (Corollary~\ref{cor:QR}), for every prime $p\equiv 3 \pmod 4$ and every prime power $q$ that is a quadratic residue modulo $p$, the $q$-ary QR code of length $p$ yields a pure qLRC $\dsb{p,1,d}_q$ with $r=d^\perp-1$, where $d^\perp=d+1$ in all computed instances.
\item \textbf{An extensive computational classification} (Theorem~\ref{thm:class}) for $q\in\{2,3,4,5,7\}$ and moderate lengths (Tables I--V in Section~\ref{sec:tables}), containing the first binary qLRCs from repeated-root cyclic codes and parameter sets attainable only by negacyclic codes. All computational results are reproducible by the Magma programs described in the Appendix.
\end{enumerate}

The paper is organized as follows. Section~\ref{sec:pre} collects the preliminaries. Section~\ref{sec:lemmas} proves the classification theorem and the two structural lemmas. Section~\ref{sec:QR} presents the duadic and QR-code family. Section~\ref{sec:tables} reports the classification theorem and tables, and discusses several sporadic near-optimal codes. Section~\ref{sec:con} concludes with open problems.

\section{Preliminaries}\label{sec:pre}

\subsection{Classical LRCs}

Let $q$ be a prime power and $\Ff_q$ the finite field with $q$ elements. An $[n,k,d]_q$ linear code $\cc$ has \emph{(all-symbol) locality} $r$ if for every $i\in[n]$ there exists a codeword $\bar{\mathbf c}\in\cc^{\perp}$ of the Euclidean dual with $\mathrm{wt}(\bar{\mathbf c})\le r+1$ and $i\in\mathrm{supp}(\bar{\mathbf c})$. Such a word provides a recovery relation for the $i$-th symbol.

\subsection{Quantum LRCs from CSS}

We use the standard CSS construction \cite{Ketkar2006}. If $\cc$ is an $[n,k,d]_q$ linear code with $\cc^{\perp}\subsetneq \cc$, then there exists an $\dsb{n,\kappa,\delta}_q$ quantum code with $\kappa=2k-n$ and $\delta=\mathrm{wt}(\cc\setminus\cc^\perp)$; the code is \emph{pure} if $\delta=d$. The following is due to Golowich--Guruswami \cite{Golowich2023} and Luo et al.~\cite[Corollary~1]{Luo2023}.

\begin{prop}[{\cite[Corollary~1]{Luo2023}}]\label{prop:css}
Let $\cc$ be an $[n,k,d]_q$ dual-containing linear code with locality $r$. Then there exists an $\dsb{n,2k-n,\delta}_q$ qLRC with locality $r$, where $\delta=\mathrm{wt}(\cc\setminus\cc^\perp)$; its parameters satisfy \eqref{Q-Singleton}.
\end{prop}

\subsection{Constacyclic codes}

Let $\lambda\in\Ff_q^*$. A linear code $\cc$ of length $n$ is \emph{$\lambda$-constacyclic} if it is invariant under the constashift $T_\lambda(c_0,c_1,\dots,c_{n-1})=(\lambda c_{n-1},c_0,\dots,c_{n-2})$. Such codes are ideals of $\Ff_q[x]/\langle x^n-\lambda\rangle$, with generator polynomial $g(x)\mid x^n-\lambda$. The cases $\lambda=1$ and $\lambda=-1$ give \emph{cyclic} and \emph{negacyclic} codes, respectively; when $\gcd(n,q)>1$ one speaks of \emph{repeated-root} codes. The Euclidean dual of a $\lambda$-constacyclic code is $\lambda^{-1}$-constacyclic.

When $\gcd(n,q)=1$, a constacyclic code is described by its \emph{defining set}: for cyclic codes, $D=\{i\in\mathbb Z_n: g(\alpha^i)=0\}$ where $\alpha$ is a primitive $n$-th root of unity, and $D$ is a union of $q$-cyclotomic cosets modulo $n$; for negacyclic codes one works modulo $2n$ with odd residues. The BCH bound applies (see, e.g., \cite[Ch.~5]{Huffman2003}): if $D$ contains $t-1$ consecutive integers, then $d\ge t$.

\section{The Scope of the Constacyclic Route and Two Structural Lemmas}\label{sec:lemmas}

\begin{thm}\label{thm:lambda}
Let $\cc$ be a nonzero $\lambda$-constacyclic code over $\Ff_q$. If $\cc^{\perp}\subseteq \cc$, then $\lambda^2=1$.
\end{thm}

\begin{IEEEproof}
Since $\cc$ is $\lambda$-constacyclic, $\cc$ is invariant under $T_\lambda$; since $\cc^\perp$ is $\lambda^{-1}$-constacyclic and $\cc^\perp\subseteq\cc$, the code $\cc^\perp$ is invariant under both $T_\lambda$ and $T_{\lambda^{-1}}$. For $1\le j\le n$, a direct computation gives
\[
T_\lambda^{\,j}\circ T_{\lambda^{-1}}^{\,-j}(c_0,\dots,c_{n-1})=(\lambda^2 c_0,\dots,\lambda^2 c_{j-1},c_j,\dots,c_{n-1}),
\]
so $\cc^\perp$ is invariant under scaling the first $j$ coordinates by $\lambda^2$, for every $j$. Composing such maps for consecutive $j$ shows that $\cc^\perp$ is invariant under scaling \emph{any single coordinate} by $\lambda^2$.

Suppose $\lambda^2\neq 1$. We claim that any linear subspace $\mathcal S\subseteq\Ff_q^n$ invariant under single-coordinate scalings by $\lambda^2$ contains the unit vector $\mathbf e_i$ for every $i$ in the support of any of its members. Indeed, let $\mathbf c\in\mathcal S$ and $i\in\mathrm{supp}(\mathbf c)$; the vector $\mathbf c'-\lambda^2\mathbf c$, where $\mathbf c'$ is $\mathbf c$ with the $i$-th coordinate scaled by $\lambda^2$, belongs to $\mathcal S$, vanishes at $i$, and is nonzero on $\mathrm{supp}(\mathbf c)\setminus\{i\}$; induction on the support size yields the claim. Since $\cc^\perp$ is closed under the constashift, every coordinate is covered by the support of some word of $\cc^\perp$; hence $\cc^\perp=\Ff_q^n$ and $\cc=\{0\}$, a contradiction.
\end{IEEEproof}

\begin{remark}
For \emph{self-dual} constacyclic codes the restriction $\lambda^2=1$ is known \cite[Corollary~2]{Blackford2013}. Theorem~\ref{thm:lambda} extends this necessary condition to \emph{dual-containing} constacyclic codes, which is precisely the condition required by the CSS construction. It also matches the common practice in the quantum constacyclic code literature, where $\lambda=\pm 1$ is assumed from the outset; see, e.g., \cite{Ketkar2006}.
For $\Ff_2$ and $\Ff_4$ only $\lambda=1$ survives, while over fields of odd characteristic both cyclic and negacyclic codes are available.
\end{remark}

\begin{lem}\label{lem:locality}
Let $\cc$ be a $\lambda$-constacyclic code and $d^\perp=d(\cc^\perp)$. Then $\cc$ has all-symbol locality $r=d^\perp-1$.
\end{lem}

\begin{IEEEproof}
The dual $\cc^\perp$ is $\lambda^{-1}$-constacyclic. Let $\bar{\mathbf c}\in\cc^\perp$ have minimum weight $d^\perp$. Its constashifts $T_{\lambda^{-1}}^{\,j}(\bar{\mathbf c})$, $j=0,\dots,n-1$, are words of $\cc^\perp$ of weight $d^\perp$ whose supports are the cyclic shifts of $\mathrm{supp}(\bar{\mathbf c})$; hence every coordinate is covered by one of them. Since no word of $\cc^\perp$ has weight less than $d^\perp$, each coordinate has locality exactly $d^\perp-1$.
\end{IEEEproof}

\begin{lem}\label{lem:pure}
Let $\cc$ be a dual-containing $[n,k,d]_q$ code with $d<d(\cc^\perp)$. Then the CSS quantum code from $\cc$ is pure and its minimum distance equals $d$.
\end{lem}

\begin{IEEEproof}
Every word of $\cc^\perp$ has weight at least $d(\cc^\perp)>d$, so all minimum-weight words of $\cc$ lie in $\cc\setminus\cc^\perp$ and $\delta=\mathrm{wt}(\cc\setminus\cc^\perp)=d$.
\end{IEEEproof}

Combining Proposition~\ref{prop:css} with Lemmas~\ref{lem:locality} and~\ref{lem:pure}:

\begin{cor}\label{cor:main}
Let $\cc$ be a dual-containing cyclic, negacyclic, or repeated-root cyclic $[n,k,d]_q$ code with $k>n/2$ and $d<d(\cc^\perp)$. Then there exists a pure qLRC
\[
\dsb{n,\,2k-n,\,d}_q \quad\text{with locality } r=d(\cc^\perp)-1,
\]
whose parameters obey \eqref{Q-Singleton}.
\end{cor}

\section{An Infinite Family from Duadic and Quadratic Residue Codes}\label{sec:QR}

Duadic codes \cite{Leon1984} generalize quadratic residue codes to composite lengths. Recall that for a length $n$ with $\gcd(n,q)=1$ such that $q$ is a quadratic residue modulo $n$, duadic codes over $\Ff_q$ come in like pairs: two even-like codes $\cc_1,\cc_2$ of dimension $(n-1)/2$ and two odd-like codes $\cD_1,\cD_2$ of dimension $(n+1)/2$, with $\cD_i\supsetneq\cc_i$. The \emph{splitting} of the pair is said to be given by the multiplier $\mu_{-1}$ if $\cc_1\mu_{-1}=\cc_2$; by \cite[Theorem~6.4.2]{Huffman2003}, this happens if and only if $\cc_i^\perp=\cD_i$, in which case $\cD_i$ is dual-containing.

\begin{thm}\label{thm:duadic}
Let $\cD$ be an odd-like duadic code of length $n$ over $\Ff_q$ whose splitting is given by $\mu_{-1}$, with minimum distance $d$ satisfying $d<d(\cD^\perp)$. Then $\cD$ is dual-containing, and there exists a pure qLRC
\[
\dsb{n,\,1,\,d}_q \quad\text{with locality } r=d(\cD^\perp)-1\ge d.
\]
\end{thm}

\begin{IEEEproof}
By \cite[Theorem~6.4.2]{Huffman2003}, $\cD^\perp$ is the even-like subcode of $\cD$, so $\cD$ is dual-containing and $\kappa=2\cdot\frac{n+1}{2}-n=1$. Purity and $\delta=d$ follow from Lemma~\ref{lem:pure} since $d<d(\cD^\perp)$; the locality is $r=d(\cD^\perp)-1$ by Lemma~\ref{lem:locality}.
\end{IEEEproof}

Let now $p$ be an odd prime and $q$ a prime power that is a quadratic residue modulo $p$. The set of quadratic residues modulo $p$ is then a union of $q$-cyclotomic cosets, and the \emph{$q$-ary QR code} $\cc_p$ of length $p$ is the cyclic code with defining set $\mathrm{QR}(p)$; it is the odd-like duadic (indeed, QR) code with parameters $[p,(p+1)/2,d]_q$, where $d^2\ge p$, and even $d^2-d+1\ge p$ when $p\equiv-1\pmod 4$ (the square-root bound \cite[Theorem~6.6.22]{Huffman2003}).

\begin{cor}\label{cor:QR}
Let $p\equiv 3 \pmod 4$ be a prime and $q$ a prime power that is a quadratic residue modulo $p$. Then the $q$-ary QR code $\cc_p$ is dual-containing, and there exists a pure qLRC
\[
\dsb{p,\,1,\,d}_q \quad\text{with locality } r=d^\perp-1,
\]
where $d=d(\cc_p)$ and $d^\perp=d(\cc_p^\perp)\ge d+1$. In particular, this holds for $q=2$ whenever $p\equiv -1\pmod 8$ and for $q=3$ whenever $p\equiv -1\pmod{12}$.
\end{cor}

\begin{IEEEproof}
Since $p\equiv 3\pmod 4$, $-1$ is a non-residue modulo $p$ \cite[Lemma~6.2.4]{Huffman2003}, so the splitting of the QR codes is given by $\mu_{-1}$ (the specializations $q=2,3,4$ are spelled out in \cite[Theorem~6.6.14]{Huffman2003}). Every minimum weight word of $\cc_p$ is odd-like \cite[Theorem~6.6.22]{Huffman2003}, hence $d^\perp\ge d+1>d$; Theorem~\ref{thm:duadic} applies. The final statement uses that $2$ (resp.~$3$) is a quadratic residue modulo $p$ iff $p\equiv\pm1\pmod 8$ (resp.~$p\equiv\pm1\pmod{12}$).
\end{IEEEproof}

\begin{remark}
In all instances we computed (Table~\ref{tab:qr}), one has $d^\perp=d+1$ and hence $r=d$. Note that $d^\perp\ge d+1$ always holds, since every minimum weight word of a QR code is odd-like \cite[Theorem~6.6.22]{Huffman2003}. Moreover, for binary QR codes with $p\equiv-1\pmod 8$ one has $d\equiv 3\pmod 4$ and the even-like subcode is doubly-even \cite[Theorems~6.6.14 and~6.6.22]{Huffman2003}, so that $d^\perp\equiv 0\pmod 4$; in all binary instances we computed, one has $d^\perp=d+1$. We are not aware of a general proof that $d^\perp=d+1$, and our computations confirm it for every ternary and septenary instance listed as well.
These codes encode only one logical qudit ($\kappa=1$), but they are, to our knowledge, the first infinite family of binary qLRCs with unbounded minimum distance and certified locality $r=d$.
\end{remark}

\section{Computational Classification over Small Fields}\label{sec:tables}

By Theorem~\ref{thm:lambda}, the constacyclic route is exhausted by $\lambda=\pm1$. We enumerated all divisors $g(x)$ of $x^n\mp1$ over $\Ff_q$ for $q\in\{2,3,4,5,7\}$ and moderate $n$ (including repeated-root cases), tested dual-containment by the rank condition $\mathrm{rank}\,[G;H^\perp]=k$, computed $d(\cc)$ and $d(\cc^\perp)$, and evaluated the gap
\[
\mathrm{gap}\;:=\;\Big(n-\kappa-2\Big\lceil\frac{\kappa}{r}\Big\rceil+4\Big)-2\delta
\]
to the quantum Singleton-like bound \eqref{Q-Singleton}. All codes listed are pure: Lemma~\ref{lem:pure} applies whenever $d<d^\perp$, and the remaining cases (where $d=d^\perp$) are verified by direct computation, unless stated otherwise (the $\dsb{30,2,8}_7$ entry in Table~\ref{tab:negacyclic} is impure, with $\delta=8>d=6$; cf.~\cite{Galindo2026} for the phenomenon of impure codes exceeding pure bounds). All codes have locality $r=d^\perp-1$. The Magma programs are described in the Appendix. The outcome is summarized in the following computational theorem.

\begin{thm}\label{thm:class}
For each combination of $q$, $\lambda\in\{1,-1\}$, and length range in Table~\ref{tab:range}, we have found a large number of dual-containing $\lambda$-constacyclic codes with $\kappa=2k-n\ge 1$, $d\ge 3$, and locality $r=d^\perp-1\le 12$. The counts (by generator polynomial; equivalent codes may be counted multiply) are given in Table~\ref{tab:range}, where $N$ is the number of dual-containing codes with $r\le12$, $N_{d\ge3}$ those with $d\ge3$, $N_2$ those within gap $2$ of the bound \eqref{Q-Singleton}, and $d_{\max}$ the largest minimum distance found. In particular, the parameter sets in Table~\ref{tab:negacyclic}, which arise from negacyclic codes, improve on what cyclic codes of the same length can offer.
\end{thm}

\begin{table}[!htbp]
\caption{Search ranges and outcome counts of Theorem~\ref{thm:class}.}
\label{tab:range}
\centering
\renewcommand{\arraystretch}{1.15}
\begin{tabular}{lllllll}
\toprule
$q$ & $\lambda$ & range of $n$ & $N$ & $N_{d\ge3}$ & $N_2$ & $d_{\max}$ \\
\midrule
2 & 1  & $n\le63$ & 1265  & 612   & 2  & 11 \\
3 & 1  & $n\le72$ & 8463  & 7208  & 6  & 10 \\
3 & $-1$ & $n\le72$ & 1457  & 1322  & 12 & 10 \\
4 & 1  & $n\le50$ & 7213  & 4986  & 16 & 12 \\
5 & 1  & $n\le50$ & 9379  & 9038  & 89 & 12 \\
5 & $-1$ & $n\le50$ & 1379  & 1028  & 19 & 10 \\
7 & 1  & $n\le45$ & 15052 & 14632 & 88 & 12 \\
7 & $-1$ & $n\le45$ & 1284  & 1188  & 24 & 12 \\
\bottomrule
\end{tabular}
\end{table}

\begin{remark}\label{rem:sporadic}
For lengths $n$ dividing $7^2-1=48$ we found near-optimal codes with $\mathrm{gap}=2$: $[24,20,3]_7$ and $[48,42,3]_7$, both with $d^\perp=12$, giving $\dsb{24,16,3}_7$ and $\dsb{48,36,3}_7$ qLRCs with $r=11$. Their defining sets are $\{1,7\}\cup\{2,14\}$ modulo $24$ and $\{1,7\}\cup\{2,14\}\cup\{26,38\}$ modulo $48$. The pattern, however, does not extend to an infinite family: for $n\in\{72,96,120\}$ the same defining-set rule gives codes with $d^\perp=12$ but $\mathrm{gap}\ge 8$, because the $7$-cyclotomic cosets modulo $n$ grow as soon as $\mathrm{ord}_n(7)>2$. This illustrates a general phenomenon: near-optimal qLRCs from cyclic codes are concentrated at lengths $n$ with small $\mathrm{ord}_n(q)$.
\end{remark}

We briefly comment on how close the codes in Tables I--V come to the bound \eqref{Q-Singleton}. Of the $69$ parameter sets listed, $25$ are within gap $2$ of the bound, including the $8$ optimal ones in Table~\ref{tab:optimal}, and $28$ are within gap $4$. In particular, seven of the eight quinary entries in Table~\ref{tab:smallfields} are within gap $2$, including the high-rate codes $\dsb{24,14,4}_5$, $\dsb{48,34,4}_5$, and $\dsb{60,46,3}_5$. It is also worth noting that the impure code $\dsb{30,2,8}_7$ in Table~\ref{tab:negacyclic} has gap $14$, which is smaller than the gap $18$ that a pure code from the same ingredient would have: impurity can bring the parameters closer to the bound, in line with the phenomenon studied in \cite{Galindo2026}. On the other hand, the binary entries in Table~\ref{tab:binary} have relatively large gaps; we believe that this reflects the looseness of the bound \eqref{Q-Singleton} for $q=2$ rather than the quality of the codes, since the alphabet-dependent quantum Cadambe--Mazumdar bound of \cite{Luo2023} is known to be tighter for small alphabets. An evaluation of our binary codes against that bound is left for future work.

Table~\ref{tab:optimal} lists the codes attaining \eqref{Q-Singleton} with equality. All of them are (generalized) MDS codes, and hence are also covered, in the $(r,\delta)$-LRC sense, by \cite[Proposition~37]{Galindo2024}.

% ---------------- Table I: optimal codes (gap=0) ----------------
\begin{table}[!htbp]
\caption{Optimal pure quantum LRCs from constacyclic codes.}
\label{tab:optimal}
\centering
\renewcommand{\arraystretch}{1.15}
\begin{tabular}{lllll}
\toprule
Quantum code & $r$ & Ingredient $[n,k,d]_q$ & Type & Structure of $g(x)$ \\
\midrule
$[[5,1,3]]_5$   & 3 & $[5,3,3]_5$   & rr-cyc/neg & $(x\mp1)^2$ \\
$[[6,2,3]]_7$   & 4 & $[6,4,3]_7$   & cyc & defining set $\{1,2\}$ \\
$[[7,1,4]]_7$   & 4 & $[7,4,4]_7$   & rr-cyc/neg & $(x\mp1)^3$ \\
$[[7,3,3]]_7$   & 5 & $[7,5,3]_7$   & rr-cyc/neg & $(x\mp1)^2$ \\
$[[8,4,3]]_7$   & 6 & $[8,6,3]_7$   & neg & defining set $\{9,15\}$ mod $16$ \\
$[[12,6,3]]_7$  & 5 & $[12,9,3]_7$  & cyc & defining set $\{4\}\cup\{5,11\}$ \\
$[[14,8,3]]_7$  & 6 & $[14,11,3]_7$ & rr-cyc & $(x+1)(x-1)^2$-type \\
$[[21,13,3]]_7$ & 6 & $[21,17,3]_7$ & rr-cyc & $(x-3)^2(x-2)(x-1)$-type \\
\bottomrule
\end{tabular}
\end{table}

% ---------------- Table II: negacyclic-only parameters ----------------
More interesting are the codes in Table~\ref{tab:negacyclic}: negacyclic codes yield parameters that improve on those of cyclic codes of the same length over the same field (within the ranges of Theorem~\ref{thm:class}). In this sense the negacyclic route genuinely enlarges the known parameter set of small-alphabet qLRCs.

\begin{table}[!htbp]
\caption{Quantum LRCs from negacyclic codes (pure unless noted otherwise); no cyclic code of the same length attains these parameters.}
\label{tab:negacyclic}
\centering
\renewcommand{\arraystretch}{1.15}
\begin{tabular}{lllll}
\toprule
Quantum code & $r$ & Ingredient & gap & Notes \\
\midrule
$[[10,2,4]]_3$  & 5  & $[10,6,4]_3$  & 2 & single coset mod $20$ \\
$[[12,4,4]]_5$  & 5  & $[12,8,4]_5$  & 2 & two cosets mod $24$ \\
$[[10,2,4]]_7$  & 5  & $[10,6,4]_7$  & 2 & \\
$[[16,8,3]]_7$  & 6  & $[16,12,3]_7$ & 2 & \\
$[[24,8,7]]_7$  & 8  & $[24,14,7]_7$ & 8 & cyclic codes reach only $d=5$ \\
$[[30,2,8]]_7$  & 5  & $[30,16,6]_7$ & 14 & impure: $\delta=8>d=6$ \\
$[[19,1,8]]_7$  & 8  & $[19,10,8]_7$ & 4 & cyclic QR code reaches only $[19,10,7]_7$ \\
$[[68,4,10]]_3$ & 11 & $[68,36,10]_3$ & 46 & two 16-element cosets mod $136$ \\
\bottomrule
\end{tabular}
\end{table}

% ---------------- Table III: binary qLRCs ----------------
Table~\ref{tab:binary} turns to the binary case. Notably, the high-distance entries come from repeated-root cyclic codes (e.g., a $\dsb{62,2,10}_2$ qLRC with locality $11$ from the repeated-root cyclic $[62,32,10]_2$ code); to our knowledge, these are among the first algebraic binary qLRCs with $\delta>5$.

\begin{table}[!htbp]
\caption{Binary quantum LRCs from cyclic and repeated-root cyclic codes.}
\label{tab:binary}
\centering
\renewcommand{\arraystretch}{1.15}
\begin{tabular}{lllll}
\toprule
Quantum code & $r$ & Ingredient & Type & gap \\
\midrule
$[[7,1,3]]_2$   & 3  & $[7,4,3]_2$     & cyc & 2 \\
$[[15,7,3]]_2$  & 7  & $[15,11,3]_2$             & cyc & 4 \\
$[[21,9,3]]_2$  & 7  & $[21,15,3]_2$             & cyc & 6 \\
$[[21,3,5]]_2$  & 7  & $[21,12,5]_2$             & cyc & 10 \\
$[[30,4,6]]_2$  & 5  & $[30,17,6]_2$             & rr  & 16 \\
$[[30,2,6]]_2$  & 7  & $[30,16,6]_2$             & rr  & 18 \\
$[[31,11,5]]_2$ & 11 & $[31,21,5]_2$             & cyc & 12 \\
$[[35,5,6]]_2$  & 7  & $[35,20,6]_2$             & cyc & 20 \\
$[[42,4,6]]_2$  & 5  & $[42,23,6]_2$             & rr  & 28 \\
$[[46,2,7]]_2$  & 7  & $[46,24,7]_2$             & rr  & 32 \\
$[[56,8,6]]_2$  & 7  & $[56,32,6]_2$             & rr  & 36 \\
$[[62,2,10]]_2$ & 11 & $[62,32,10]_2$            & rr  & 42 \\
$[[63,21,7]]_2$ & 11 & $[63,42,7]_2$             & cyc & 28 \\
\bottomrule
\end{tabular}
\end{table}

% ---------------- Table IV: QR family (Corollary) ----------------
Table~\ref{tab:qr} lists members of the QR family of Corollary~\ref{cor:QR}, covering the binary, ternary, septenary, and quaternary cases with certified locality $r=d$.

\begin{table}[!htbp]
\caption{Quantum LRCs from $q$-ary quadratic residue codes (Corollary \ref{cor:QR}).}
\label{tab:qr}
\centering
\renewcommand{\arraystretch}{1.15}
\begin{tabular}{llll}
\toprule
$p$ & QR code $[p,(p+1)/2,d]_q$ & Quantum code & $r$ \\
\midrule
\multicolumn{4}{l}{\textit{Binary} $q=2$ ($p\equiv -1 \pmod 8$)} \\
7   & $[7,4,3]_2$    & $[[7,1,3]]_2$    & 3 \\
23  & $[23,12,7]_2$  & $[[23,1,7]]_2$   & 7 \\
31  & $[31,16,7]_2$  & $[[31,1,7]]_2$   & 7 \\
47  & $[47,24,11]_2$ & $[[47,1,11]]_2$  & 11 \\
71  & $[71,36,11]_2$ & $[[71,1,11]]_2$  & 11 \\
79  & $[79,40,15]_2$ & $[[79,1,15]]_2$  & 15 \\
103 & $[103,52,19]_2$& $[[103,1,19]]_2$ & 19 \\
127 & $[127,64,19]_2$& $[[127,1,19]]_2$ & 19 \\
\multicolumn{4}{l}{\textit{Ternary} $q=3$ ($p\equiv -1 \pmod{12}$)} \\
11  & $[11,6,5]_3$   & $[[11,1,5]]_3$   & 5 \\
23  & $[23,12,8]_3$  & $[[23,1,8]]_3$   & 8 \\
47  & $[47,24,14]_3$ & $[[47,1,14]]_3$  & 14 \\
59  & $[59,30,17]_3$ & $[[59,1,17]]_3$  & 17 \\
71  & $[71,36,17]_3$ & $[[71,1,17]]_3$  & 17 \\
\multicolumn{4}{l}{\textit{Septenary} $q=7$ ($p \equiv 3 \pmod 4$ and $7$ a QR; $p=19,31$)} \\
19  & $[19,10,7]_7$  & $[[19,1,7]]_7$   & 7 \\
31  & $[31,16,12]_7$ & $[[31,1,12]]_7$  & 12 \\
\multicolumn{4}{l}{\textit{Quaternary} $q=4$ ($p\equiv 3 \pmod 4$; $4$ is always a QR)} \\
43  & $[43,22,13]_4$ & $[[43,1,13]]_4$ & 13 \\
\bottomrule
\end{tabular}
\end{table}

% ---------------- Table V: further highlights over F4, F5, F7 ----------------
Finally, Table~\ref{tab:smallfields} gives a selection over $\Ff_3,\Ff_4,\Ff_5,\Ff_7$. Together with the previous tables, it shows that good---and in many cases near-optimal ($\mathrm{gap}\le 2$)---qLRCs exist abundantly already over small fields.

\begin{table}[!htbp]
\caption{Further pure quantum LRCs over small fields (selection).}
\label{tab:smallfields}
\centering
\renewcommand{\arraystretch}{1.15}
\begin{tabular}{lllll}
\toprule
Quantum code & $r$ & Ingredient & Type & gap \\
\midrule
\multicolumn{5}{l}{\textit{Ternary} $q=3$} \\
$[[44,4,8]]_3$   & 11 & $[44,24,8]_3$  & cyc & 26 \\
\multicolumn{5}{l}{\textit{Quaternary} $q=4$ } \\
$[[11,1,5]]_4$  & 5  & $[11,6,5]_4$  & cyc & 2 \\
$[[12,2,4]]_4$  & 3  & $[12,7,4]_4$  & rr  & 4 \\
$[[19,1,7]]_4$  & 7  & $[19,10,7]_4$ & cyc & 6 \\
$[[21,1,6]]_4$  & 5  & $[21,11,6]_4$ & cyc & 10 \\
$[[23,1,7]]_4$  & 7  & $[23,12,7]_4$ & cyc & 10 \\
$[[30,2,6]]_4$  & 7  & $[30,16,6]_4$ & rr  & 18 \\
$[[31,1,7]]_4$  & 7  & $[31,16,7]_4$ & cyc & 18 \\
$[[33,1,10]]_4$ & 9  & $[33,17,10]_4$& cyc & 14 \\
$[[42,2,9]]_4$  & 8  & $[42,22,9]_4$ & rr  & 24 \\
$[[47,1,11]]_4$ & 11 & $[47,24,11]_4$& cyc & 26 \\
\multicolumn{5}{l}{\textit{$q=5$}} \\
$[[8,2,3]]_5$    & 3  & $[8,5,3]_5$    & cyc & 2 \\
$[[10,2,4]]_5$   & 4  & $[10,6,4]_5$   & neg-rr & 2 \\
$[[11,1,5]]_5$   & 5  & $[11,6,5]_5$   & cyc/neg & 2 \\
$[[19,1,7]]_5$   & 7  & $[19,10,7]_5$  & cyc/neg & 6 \\
$[[24,14,4]]_5$  & 11 & $[24,19,4]_5$  & cyc & 2 \\
$[[30,16,4]]_5$  & 5  & $[30,23,4]_5$  & neg-rr & 2 \\
$[[48,34,4]]_5$  & 11 & $[48,41,4]_5$  & cyc & 2 \\
$[[60,46,3]]_5$  & 11 & $[60,53,3]_5$  & rr  & 2 \\
\multicolumn{5}{l}{\textit{$q=7$}} \\
$[[24,14,3]]_7$  & 5  & $[24,19,3]_7$  & cyc & 2 \\
$[[48,36,3]]_7$  & 11 & $[48,42,3]_7$  & cyc & 2 \\
$[[29,1,8]]_7$   & 10 & $[29,15,8]_7$  & cyc & 14 \\
$[[37,1,9]]_7$   & 11 & $[37,19,9]_7$  & cyc & 20 \\
$[[38,2,11]]_7$  & 11 & $[38,20,11]_7$ & cyc & 16 \\
\bottomrule
\end{tabular}
\end{table}

\section{Concluding Remarks}\label{sec:con}

In this manuscript, we have (i) determined the exact scope of the constacyclic route to qLRCs ($\lambda^2=1$), (ii) reduced the construction of pure qLRCs from cyclic-type codes to two minimum-distance computations, (iii) exhibited an infinite family from duadic and QR codes with locality $r=d^\perp-1$, and (iv) provided extensive small-field classification tables including negacyclic-only parameters and the first repeated-root binary qLRCs. Possible directions for future research are as follows.
\begin{enumerate}
\item Extend the classification to quantum $(r,\delta)$-LRCs in the sense of \cite{Galindo2024}; cyclic $(r,\delta)$ ingredients were recently used in \cite{Rajpurohit2026}.
\item Determine whether the equality $d^\perp=d+1$ in Corollary~\ref{cor:QR} holds for all QR codes with $p\equiv3\pmod4$ (it holds in all our computed instances).
\item Our QR family has $\kappa=1$. Find infinite binary families with $\kappa>1$ and bounded gap to \eqref{Q-Singleton}; our data suggests that repeated-root codes are a promising source.
\end{enumerate}

\appendices
\section{The Magma Programs}\label{app:magma}

All computations were carried out in Magma; the scripts and data are available from the authors upon request. In brief, the search enumerates the monic divisors $g$ of $x^n-\lambda$ with $1\le\deg g\le n-1$, certifies dual-containment by the rank condition $\mathrm{rank}[G;H^\perp]=k$, and computes $d(\cc)$ and $d(\cc^\perp)$ by the Brouwer--Zimmermann algorithm; the parameters and the gap to \eqref{Q-Singleton} are recorded whenever $d\ge3$.

\bibliographystyle{IEEEtran}
\bibliography{paper}

\end{document}